\documentclass[11pt]{article}

\usepackage[nottoc]{tocbibind}

\usepackage{amsfonts}
\usepackage{amssymb}
\usepackage{amstext}
\usepackage{amsmath}
\usepackage{mathtools}

\usepackage{amsthm} 

\usepackage{caption}
\usepackage{subcaption}

\usepackage{color}
\usepackage{nameref}
\definecolor{ForestGreen}{rgb}{0.1333,0.5451,0.1333}
\definecolor{DarkRed}{rgb}{0.8,0,0}
\definecolor{Red}{rgb}{0.9,0,0}
\usepackage[linktocpage=true,
	pagebackref=true,colorlinks,
	linkcolor=DarkRed,citecolor=ForestGreen,
	bookmarks,bookmarksopen,bookmarksnumbered]
	{hyperref}
\usepackage{cleveref}

\usepackage{thmtools,thm-restate} 

\usepackage[numbers,sort&compress]{natbib}

\usepackage{graphicx}
\usepackage{graphics}
\usepackage{colordvi}
\usepackage{xspace}
\usepackage{algorithm}
\usepackage{algorithmicx}
\usepackage{url}
\usepackage{enumitem}

        {\hspace*{\fill}$\Box$\par\vspace{4mm}}

\usepackage[left=1in,top=1in,right=1in,bottom=1in]{geometry} 
\usepackage{mathptmx} 
\usepackage[scaled=.90]{helvet}
\usepackage{courier}
\usepackage{booktabs}
\usepackage{threeparttable}
\usepackage[small,compact]{titlesec}
\setlist[itemize]{noitemsep,nolistsep} 
\setlist[enumerate]{noitemsep,nolistsep} 

\makeatletter
\renewcommand{\paragraph}{%
  \@startsection{paragraph}{4}%
  {\z@}{1ex \@plus 1ex \@minus .2ex}{-1em}%
  {\normalfont\normalsize\bfseries}%
}
\makeatother

\makeatletter
\def\thmt@refnamewithcomma #1#2#3,#4,#5\@nil{%
  \@xa\def\csname\thmt@envname #1utorefname\endcsname{#3}%
  \ifcsname #2refname\endcsname
    \csname #2refname\expandafter\endcsname\expandafter{\thmt@envname}{#3}{#4}%
  \fi
}
\makeatother

\declaretheorem[numberwithin=section,refname={Theorem,Theorems},Refname={Theorem,Theorems}]{theorem}
\declaretheorem[numberlike=theorem,refname={Lemma,Lemmas},Refname={Lemma,Lemmas}]{lemma}

\declaretheorem[numberlike=theorem,refname={Corollary,Corollaries},Refname={Corollary,Corollaries}]{corollary}

\declaretheorem[numberlike=theorem,refname={property,properties},Refname={Property,Properties}]{property}

\declaretheorem[numberlike=theorem,refname={Claim, Claims},Refname={Claim, Claims}]{claim}

\declaretheorem[numberlike=theorem]{definition}

\newtheorem{invariant}[theorem]{Invariant}

\newcommand{\ceil}[1]{\ensuremath{\left\lceil#1\right\rceil}}

\renewcommand{\phi}{\varphi}

\newcommand{\poly}{\operatorname{poly}}

\renewcommand{\t}{\text{slack}}

\newcommand{\eps}{\delta}

\newcommand{\G}{\mathcal{G}}

\ifdefined\ShowComment

\def\danupon#1{\marginpar{$\leftarrow$\fbox{D}}\footnote{$\Rightarrow$~{\sf #1 --Danupon}}}
\def\sayan#1{\marginpar{$\leftarrow$\fbox{S}}\footnote{$\Rightarrow$~{\sf #1 --Sayan}}}
\def\monika#1{\marginpar{$\leftarrow$\fbox{M}}\footnote{$\Rightarrow$~{\sf #1 --Monika}}}

\else

\def\danupon#1{}
\def\sayan#1{}
\def\monika#1{}

\fi

\newboolean{short}
\setboolean{short}{true} 

\newcommand{\shortOnly}[1]{\ifthenelse{\boolean{short}}{#1}{}}
\newcommand{\longOnly}[1]{\ifthenelse{\boolean{short}}{}{#1}}

\title{Deterministic Fully Dynamic Approximate Vertex Cover and Fractional Matching in $O(1)$ Amortized  Update Time}

\date{}

\author{Sayan Bhattacharya\thanks{Institute of Mathematical Sciences, Chennai, India. Email: {\tt bsayan@imsc.res.in}}  \and Deeparnab Chakrabarty\thanks{Micorsoft Research, India. Email: {\tt dechakr@microsoft.com}} \and Monika Henzinger\thanks{University of Vienna, Austria. Email: {\tt monika.henzinger@univie.ac.at}.The research leading to these results has received funding from the European Research Council under the European Union's Seventh Framework Programme (FP/2007-2013) / ERC Grant Agreement no. 340506.}}

\begin{document}

\setcounter{tocdepth}{3}

\begin{titlepage}
\maketitle
\pagenumbering{roman}

\begin{abstract}
We consider the problems of maintaining an approximate maximum matching and an approximate minimum vertex cover in a dynamic graph undergoing a sequence of edge insertions/deletions. Starting with the seminal work of Onak and Rubinfeld [STOC 2010], this problem has received significant attention in recent years. Very recently, extending the framework of Baswana, Gupta and Sen [FOCS 2011], Solomon [FOCS 2016] gave a randomized dynamic algorithm for this problem that has an approximation ratio of $2$ and an amortised update time of $O(1)$ with high probability. This algorithm requires the assumption of an {\em oblivious adversary}, meaning that the future sequence of edge insertions/deletions in the graph cannot depend in any way on the algorithm's past output.  A natural way to remove the assumption on oblivious adversary is to give a deterministic dynamic algorithm for the same problem in $O(1)$ update time. In this paper, we resolve this question. 

We present a new {\em deterministic} fully dynamic algorithm that  maintains a $O(1)$-approximate minimum vertex cover and maximum fractional matching, with an amortised update time of $O(1)$. Previously, the best  deterministic algorithm for this problem was due to Bhattacharya, Henzinger and Italiano [SODA 2015]; it had an approximation ratio of $(2+\epsilon)$ and an amortised update time of $O(\log n/\epsilon^2)$. 
Our results also extend to a fully dynamic $O(f^3)$-approximate algorithm with $O(f^2)$ amortized update time for the hypergraph vertex cover and fractional hypergraph matching problem where every hyperedge has at most $f$ vertices.
\end{abstract}

\newpage
\pagenumbering{gobble}
\clearpage

\setcounter{tocdepth}{3}


\end{titlepage}

\newpage

\pagenumbering{arabic}

\newcommand{\DD}{{\sc Down-Dirty }}
\newcommand{\UD}{{\sc Up-Dirty }}
\newcommand{\Dplus}{{\sc $+$-Dirty }}
\newcommand{\Clean}{{\sc Super-Clean} }

\section{Introduction}
\label{mainbody:sec:intro}
Computing a maximum cardinality matching is a fundamental problem in computer science with applications, for example, in operations research, computer science, and computational chemistry. In many of these applications the underlying graph can change. Thus, it is natural to ask how quickly a maximum matching can be maintained after a change in the graph. 
As nodes usually change less frequently than edges, dynamic matching algorithms usually study the problem where edges are inserted and deleted, which is called the {\em (fully) dynamic matching problem}\footnote{Node updates are usually handled through the insertion and deletion of
isolated nodes, but there has been also some work on the node insertions-only or node deletions-only problem~\cite{BosekLSZ14}.}. 
The goal of a dynamic matching algorithm is to maintain either an actual matching (called the {\em matching version}) or the {\em value} of the matching 
(called the {\em value version}) as efficiently as possible. 

Unfortunately, the problem of maintaining even just the value of the maximum cardinality matching is hard: There is a conditional lower bound that
shows that no (deterministic or randomized) algorithm can achieve at the same time an amortized update time of $O(m^{1/2-\epsilon})$ and a query (for the size of the matching) time of $O(m^{1-\epsilon})$ for any small $\epsilon >0$~\cite{HenzingerKNS15}
(see~\cite{AbboudW14} for conditional lower bounds using different assumptions). The best upper bound is Sankowski's randomized algorithm~\cite{Sankowski07} that solves the value problem in time $O(n^{1.495})$ per update and $O(1)$ per query. Thus, it is natural to study the dynamic {\em approximate} maximum matching problem, and there has been a large body~\cite{OnakR10,BaswanaGS11,NeimanS13, GuptaP13, BhattacharyaHI15s,BhattacharyaHN16,Solomon16} of work on it and its dual, the approximate vertex cover problem, in the last few years. 

Dynamic algorithms can be further classified  into two types: Algorithms that require an {\em oblivious} (aka {\em non-adaptive}) adversary, i.e., an adversary that does {\em not} base future updates and queries on the answers 
to past queries, and algorithms that work even for an adaptive adersary. Obviously, the earlier kind of algorithms are less general than the later. Unfortunately, all randomized dynamic approximate matching and vertex cover algorithms so 
far either are not known to work with an adaptive adversary~\cite{OnakR10} or do not work for
an adaptive adversary~\cite{BaswanaGS11,Solomon16}. Solomon~\cite{Solomon16} gives the best such randomized algorithm: It achieves $O(1)$ amortized update time  (with high probability) and $O(1)$ query time for maintaining a $2$-approximate maximum matching and a $2$-approximate minimum vertex cover. He also extends this result to the dynamic distributed setting (\`a la Parter, Peleg, and Solomon~\cite{ParterPS16}) with the same approximation ratio and update cost. 

In this paper we present the first {\em deterministic} algorithm that maintains an $O(1)$-approximation to the size of the maximum matching in
$O(1)$ amortized update time and $O(1)$ query time. We also maintain an $O(1)$-approximate vertex cover in the same update time.
Note that this is the first {\em deterministic} dynamic algorithm with constant update time for any non-trivial dynamic  graph problem. 
This is significant as for other dynamic problems such as the dynamic connectivity problem or the dynamic planarity testing problem there are non-constant lower bounds in the cell probe model on the time per operation~\cite{HenzingerF98,Patrascu16}. Thus, our result shows that no such lower bound can exist for the dynamic approximate matching problem.

There has been prior work on deterministic algorithms for dynamic approximate matching, but they all have  $\Omega(\poly (\log n))$ update time: One line of work concentrated on reducing the approximation ratio as much as possible, or
at least below 2: Neiman and Solomon~\cite{NeimanS13} achieved an  update time $O(\sqrt{m})$ for maintaining a $3/2$-approximate maximum matching and $2$-approximate minimum vertex cover. This result was improved by Gupta and Peng~\cite{GuptaP13} that gave an algorithm with update time $O(\sqrt{m}/\epsilon^2)$ for maintaining a $(1+\epsilon)$-approximate maximum matching.  
Recently, Bernstein and Stein~\cite{BernsteinS16} gave an algorithm with $O(m^{1/4}/\epsilon^2)$ amortised update time for maintaining a $(3/2+\epsilon)$-approximate maximum matching.  Another line of work, and this paper fits in this line, concentrated on getting a constant approximation while reducing the update time to polylogarithmic: Bhattacharya, Henzinger and Italiano~\cite{BhattacharyaHI15s} achieved an $O(\log n/\epsilon^2)$ update time for maintaining a $(2+\epsilon)$-approximate maximum {\em fractional} matching and a $(2+\epsilon)$-approximate minimum vertex cover. Note that any
fractional matching algorithm solves the value version of the dynamic matching problem while degrading the approximation ratio by a factor of $3/2$. Thus, the algorithm in~\cite{BhattacharyaHI15s}  maintains a $(3+\epsilon)$-approximation of the value of the maximum matching. The fractional matching in this algorithm was later ``determinically rounded'' by  Bhattacharya, Henzinger and Nanongkai~\cite{BhattacharyaHN16} to achieve a $O(\text{poly} (\log n, 1/\epsilon))$ update time for maintaining a $(2+\epsilon)$-approximate maximum matching. 

Our method also generalizes to the hypergraph vertex (set) cover and hypergraph fractional matching problem which was considered by ~\cite{BhattacharyaHI15}. In this problem the hyperedges of a hypergraph are inserted and deleted over time. $f$ indicates the maximum cardinality of any hyperedge. The objective is to maintain a hypergraph vertex cover, that is, a set of vertices that hit every hyperedge. Similarly a fractional matching in the hypergraph is a fractional assignment (weights) to the hyperedges so that the total weight faced by any vertex is at most $1$.
We give an $O(f^3)$-approximate algorithm with amortized $O(f^2)$ update time.

\subsection{Our Techniques}
Our algorithm builds and simplifies the framework of hierarchical partitioning of vertices proposed by Onak and Rubinfeld~\cite{OnakR10}, which was later enhanced by Bhattacharya, Henzinger and Italiano~\cite{BhattacharyaHI15s} to give a deterministic fully-dynamic $(2+\eps)$-approximate vertex cover and maximum matching in $O(\log n/\eps^2)$-amortized update time. The hierarchical partition divides the vertices into $O(\log n)$-many levels and maintains a {\em fractional matching} and vertex cover. To prove that the approximation factor is good, Bhattacharya et. al.\cite{BhattacharyaHI15s} also maintain approximate complementary slackness conditions. An edge insertion or deletion can disrupt these conditions (and indeed at times the feasibility of the fractional matching), and a {\em fixing procedure} maintains various invariants. To argue that the update time is bounded,~\cite{BhattacharyaHI15s} give a rather involved potential function argument which proves that the update time bounded by $O(L)$, the number of levels, and is thus $O(\log n)$. It seems unclear whether the update time can be argued to be a constant or not.

Our algorithm is morally similar  to that in Bhattacharya et. al.\cite{BhattacharyaHI15s}, except we are a bit stricter when we fix nodes. As in~\cite{BhattacharyaHI15s}, whenever an edge insertion or deletion or a previous update violates an invariant condition, we move nodes across the partitioning (incurring update costs), but after a node is fixed we often ensure it satisfies a stronger condition than what the invariant requires. For example, suppose a node $v$ violates the upper bound of a fractional matching, that is, the total fractional weight it faces becomes larger than $1$, then the fixing subroutine will at the end ensure that the final weight the node faces is significantly less than $1$. Morally, this slack allows us to make an charging argument of the following form -- if this node violates the upper bound again, then a lot of ``other things'' must have occurred to increase its weight (for instance, maybe edge insertions have occurred). Such a charging argument, essentially, allows us to bypass the $O(\log n)$-update time to an $O(1)$-update time. The flip side of the slack is that our complementary slackness conditions become weak, and therefore instead of a $2+\varepsilon$-approximation we can only ensure an $O(1)$-approximation. The same technique easily generalizes to the hypergraph setting.
It would be interesting to see other scenarios where approximation ratios can be slightly traded in for huge improvements in the update time.

\paragraph{Remark.}  Very recently, and independently of our work, Gupta et al.~\cite{Gupta} obtained a $O(f^3)$-approximation algorithm for maximum fractional matching and minimum vertex cover in a hyper-graph in $O(f^2)$ amortized update time.

\section{Notations and Preliminaries}
\label{mainbody:sec:prelim}

Since the hypergraph result implies the graph result, henceforth we consider the former problem. 
The input hyper-graph $\mathcal{G} = (V, E)$ has $|V| = n$ nodes. Initially, the set of hyper-egdes  is empty, i.e., $E = \emptyset$. Subsequently, an adversary inserts or deletes  hyper-edges in the hyper-graph $\mathcal{G} = (V, E)$. The node-set $V$ remains unchanged with time.  Each hyper-edge  contains at most $f$ nodes. We say that $f$ is the maximum {\em frequency} of a hyper-edge. If a hyper-edge $e$ has a node $v$ as one of its endpoints, then we write $v \in e$. For every node $v \in V$, we let $E_v = \{ e \in E : v \in e \}$ denote the set of hyper-edges that are incident on $v$. 
In this fully dynamic setting, our goal is to maintain an approximate maximum fractional matching and  an approximate minimum vertex cover in $\mathcal{G}$. The main result of this paper is summarized in the theorem below.

\begin{theorem}
\label{th:main}
We can maintain an $O(f^3)$ approximate maximum fractional matching and an $O(f^3)$ approximate minimum vertex cover in the input hyper-graph $\mathcal{G} = (V, E)$ in $O(f^2)$ amortized update time.
\end{theorem}

To prove Theorem~\ref{th:main}, throughout the rest of this paper we fix  two parameters  $\alpha$ and $\beta$ as follows.
\begin{equation}
\label{eq:parameter}
\beta = 17, \text{ and } \alpha = 1 + 36 f^2 \beta^2.
\end{equation}

We will   maintain a hierarchical partition of the node-set $V$ into $L +1$ levels $\{0, \ldots, L\}$, where $L = \ceil{f \cdot \log_{\beta} n}+1$. We let $\ell(v) \in \{0, \ldots, L\}$ denote the {\em level} of a node $v \in V$. We define the {\em level} of a hyper-edge $e \in E$ to be  the maximum level among its endpoints, i.e., $\ell(e) = \max_{v \in e}\ell(v)$.  
The levels of nodes (and therefore hyper-edges) induce the following weights on hyper-edges: $w(e) := \beta^{-\ell(e)}$ for every hyper-edge $e \in E$. 
For all nodes $v \in V$, let $W_v := \sum_{e \in E_v} w(e)$ be the total weight received by $v$ from its incident hyper-edges.  We will satisfy the following invariant after  processing a hyper-edge insertion or deletion.

 \begin{invariant}
 \label{mainbody:inv:node}
 Every node $v \in V$ at level $\ell(v) > 0$ has weight $1/(\alpha \beta^2) < W_v < 1$. Every node $v \in V$ at level $\ell(v) = 0$ has weight $0 \leq W_v \leq 1/\beta^2$. 
 \end{invariant}

 \begin{corollary}
 \label{mainbody:cor:edge}
Under Invariant~\ref{mainbody:inv:node}, the nodes in levels $\{1, \ldots, L\}$ form a vertex cover in $\mathcal{G}$.
 \end{corollary}
 
 \begin{proof}
 Suppose that there is a hyper-edge $e \in E$ with $\ell(v) = 0$ for all $v \in e$. Then we also have $\ell(e) = 0$ and   $w(e) = 1/\beta^{\ell(e)} = 1/\beta^0 = 1$. So for every node $v \in e$, we get: $W_v \geq w(e) = 1$. This violates Invariant~\ref{mainbody:inv:node}. 
 \end{proof}
 \noindent
Invariant~\ref{mainbody:inv:node} ensures that $w(e)$'s form a fractional matching satisfying approximate complementary slackness conditions with the vertex cover defined in Corollary~\ref{mainbody:cor:edge}. This gives the following theorem. 
 \begin{theorem}
 \label{mainbody:th:old:result}
In our algorithm, the hyper-edge weights $\{ w(e) \}$ form a $f\alpha\beta^2$-approximate maximum fractional matching, and the nodes in levels $\{1, \ldots, L \}$ form a $f\alpha\beta^2$-approximate minimum vertex cover.
 \end{theorem}
\begin{proof}(Sketch)
Say that a fractional matching, which assigns a weight $w(e)$ to every hyper-edge $e \in E$, is {\em maximal} iff for every hyper-edge $e \in E$ there is some node $v \in e$ such that $W_v = 1$. Let $T = \{ v \in V : W_v = 1\}$ be the set of all {\em tight} nodes in this fractional matching. Clearly, the set of nodes $T$ form a vertex cover in $\G$. It is well known that the sizes of such a fractional matching $\{w(e) \}$ and the corresponding vertex cover $T$ are within a factor $f$ of each other. The key observation is that under Invariant~\ref{mainbody:inv:node}, the fractional matching $\{w(e)\}$ is {\em approximately} maximal, meaning that for every hyper-edge $e \in E$ there is some node $v \in e$ such that $W_v > 1/(\alpha \beta^2)$. Further, the set of nodes in levels $\{1, \ldots, L\}$ are {\em approximately} tight, since each of them has weight at least $1/(\alpha \beta^2)$. 
\end{proof}
We introduce some more notations. 
For any vertex $v$, let  $W_v^+ := \sum_{e \in E_v : \ell(e) > \ell(v)} w(e)$ be the total {\em up-weight} received by $v$, that is, weight from those incident hyper-edges whose  levels are strictly greater than $\ell(v)$. For all levels $i \in \{0,1,\ldots, L\}$, we let $W_{v\to i}$ and $W_{v\to i}^+$ respectively denote the values of $W_v$ and $W_v^+$ {\em if the node $v$ were to go to level $i$ and the levels of all the other nodes were to remain unchanged}. More precisely, for every hyper-edge $e \in E$ and node $v \in e$, we define $\ell_v(e) = \max_{u \in e : u \neq v} \ell(u)$ to be the maximum level among the endpoints of $e$ that are distinct from $v$. Then we have: $W_{v\to i} := \sum_{e \in E_v} \beta^{-\max(\ell_v(e), i)}$ and $W_{v\to i}^+ := \sum_{e \in E_v : \ell_v(e) > i} \beta^{-\ell_v(e)}$. Our algorithm maintains a notion of time such that in each time step the algorithm performs {\em one} elementary operation. We let $W_v(t)$ denote the weight (resp, up-weight) faced by $v$ {\em right before} the operation at time $t$.  Similarly define 
$W_{v\to i}(t), W^+_v(t)$, and $W^+_{v\to i}(t)$.

\paragraph{Different states of a node.}
Before the insertion/deletion of a hyper-edge in $\mathcal{G}$, all nodes satisfy Invariant~\ref{mainbody:inv:node}. When a hyper-edge  is inserted (resp. deleted), it increases (resp. decreases) the weights faced by its endpoints. Accordingly, one or more endpoints can violate Invariant~\ref{mainbody:inv:node} after the insertion/deletion of a hyper-edge. Our algorithm {\em fixes} these nodes 
by changing their levels, which may lead to new violations, and so on and so forth. To describe the algorithm, we need to define certain states of the nodes.
\begin{definition}
	\label{mainbody:def:down:dirty}
	A node $v \in V$ is \DD iff $\ell(v) > 0$ and $W_v \leq 1/(\alpha \beta^2)$. 	A node $v \in V$ is \UD iff either \{$\ell(v) = 0, W_v > 1/\beta^2$\} or \{$\ell(v) > 0, W_v \geq 1$\}.
	A node is {\sc Dirty} if it is either \DD or \UD.
\end{definition}
\noindent
Invariant~\ref{mainbody:inv:node} is satisfied if and only if no node is {\sc Dirty}. We need another definition of \Clean nodes which will be crucial.
 
\begin{definition}
\label{mainbody:def:clean} A node $v \in V$ is \Clean iff one of the following conditions hold: 
 (1) We have $\ell(v) = 0$ and  $W_v \leq 1/\beta^2$, or (2) We have $\ell(v) > 0$,   $1/\beta^2 < W_v \leq 1/\beta$, and $W_v^+ \leq 1/\beta^2$. 
 \end{definition}
 \noindent
 Note that a \Clean node $v$ with $\ell(v) > 0$ has a stronger upper bound on the weight $W_v$ it faces and also an even stronger upper bound on the up-weight $W^+_v$ it faces. At a high level, one of our subroutines will lead to \Clean nodes, and the {\em slack} in the parameters is what precisely allows us to perform an amortized analysis in the update time.

\paragraph{Data Structures.} 
 For all nodes $v \in V$ and levels $i \in \{0,1,\ldots,L\}$, let $E_{v, i} := \{ e \in E_v : \ell(e) = i\}$ denote the set of hyper-edges incident on $v$ that are at level $i$. Note that $E_{v, i} = \emptyset$ for all $i < \ell(v)$.  
 We will maintain the following data structures. 
 (1) For every level $i\in \{0,1,\ldots,L\}$ and node $v \in V$, we store the set of hyper-edges $E_{v, i}$ as a doubly linked list, and also maintain a counter that stores the number of hyper-edges in $E_{v, i}$. (2) For every node $v \in V$, we store the weights $W_v$ and $W_v^+$, its level $\ell(v)$ and an indicator variable for each of the states defined above. (3) For each hyper-edge $e \in E$, we store the values of  its level $\ell(e)$ and therefore its weight $w(e)$. Finally, using appropriate pointers, we ensure that a hyper-edge can be inserted into or deleted from any given linked list in constant time.  
 We now state two lemmas that will be useful in analysing the update time of our algorithm.
 \begin{lemma}
 	\label{lm:time:move:up}
 	Suppose that a node $v$ is currently at level $\ell(v) = i \in [0, L-1]$ and we want to move it  to some level $j \in [i+1, L]$. Then it takes $O(f \cdot |\{ e \in E_v : \ell_v(e) < j\}|)$ time to update the relevant data structures. 
 \end{lemma}
 
 \begin{proof}
 	If a hyper-edge  $e$ is not incident on the node $v$, then the data structures associated with $e$ are not affected as $v$ moves up from level $i$ to level $j$. Further, among the hyper-edges $e \in E_v$, only the ones with  $\ell_v(e) < j$  get affected (i.e., the data structures associated with them need to be changed) as $v$ moves up from level $i$ to level $j$. Finally, for every hyper-edge that gets affected, we need to spend $O(f)$ time to update the data structures for its $f$ endpoints. 
 \end{proof}
 
 \begin{lemma}
 	\label{lm:time:move:down}
 	Suppose that a node $v$ is currently at level $\ell(v) = i \in [1,L]$ and we want to move it down to some level $j \in [0,i-1]$. Then it takes $O(f \cdot |\{ e \in E_v : \ell_v(e) \leq i \} |)$ time to update the relevant data structures. 
 \end{lemma}
 
 \begin{proof} 
 	If a hyper-edge $e$ is not adjacent to the node $v$, then the data structures associated with $e$ are not affected as $v$ moves down from level $i$ to level $j$. Further, among the hyper-edges  $e \in E_v$, only the ones with $\ell_v(e) \leq i$ get affected (i.e., the data structures associated with them need to be changed) as $v$ moves down from level $i$ to level $j$. Finally, for every hyper-edge that gets affected, we need to spend $O(f)$ time to update the data structures for its $f$ endpoints. 
 \end{proof}

\section{The algorithm: Handling the insertion/deletion of a hyper-edge in the input graph}
\label{mainbody:sec:algo}

Initially, the  graph $\mathcal{G}$ is empty, every node is at level $0$, and  Invariant~\ref{mainbody:inv:node} holds.  By induction,  we will ensure that the following property is satisfied just before the insertion/deletion of a hyper-edge. 
 \begin{property}
 \label{mainbody:prop:inv}
 No node $v \in V$ is {\sc Dirty}.
 \end{property}

 \noindent{\bf Insertion of a hyper-edge $e$.} When a hyper-edge $e$ is inserted into the input graph, it is assigned a   level $\ell(e) = \max_{v \in e}\ell(v)$ and a
     weight $w(e) = \beta^{-\ell(e)}$. The hyper-edge gets inserted into the linked lists $E_{v, \ell(e)}$ for all nodes $v \in e$. Furthermore, for every node $v \in e$, the weights $W_v$  increases by $w(e)$. For every endpoint $v \in e$, if $\ell(v) < \ell(e)$, then the weight $W_v^+$ increases by $w(e)$.  As a result of these operations, one or more endpoints of $e$ can now become {\sc Up-Dirty}
     and Property~\ref{mainbody:prop:inv} might no longer be satisfied. Hence, in order to restore Property~\ref{mainbody:prop:inv} we  call the subroutine described in Figure~\ref{mainbody:fig:fix:dirty}.
 
 \smallskip
 \noindent{\bf Deletion of a hyper-edge $e$.} When a hyper-edge $e$ is deleted from the input graph, we erase all the data structures associated with it. We remove the hyper-edge from the linked lists $E_{v, \ell(e)}$ for all $v \in e$, and erase the values $w(e)$ and $\ell(e)$. For every node $v \in e$,  the weight $W_v$ decreases by $w(e)$. Further, for every endpoint $v \in e$, if $\ell(v) < \ell(e)$, then we decrease the weight $W_v^+$ by $w(e)$. As a result of these operations, one or more endpoints of $e$ can now become {\sc Down-Dirty}, and Property~\ref{mainbody:prop:inv} might get violated. Hence, in order to restore Property~\ref{mainbody:prop:inv} we  call the subroutine described in Figure~\ref{mainbody:fig:fix:dirty}.
 
  \begin{figure}[h!]
  	\centerline{\framebox{
  			\begin{minipage}{5.5in}
  				\begin{tabbing}
  					01. \ \ \ \ \ \ \ \=  {\sc While} the set of {\sc Dirty} nodes is nonempty \\
  					02. \> \ \ \ \ \ \ \ \=  {\sc If} there exists some {\sc Up-Dirty} node $v$: \\
  					03. \> \>  \ \ \ \ \ \ \ \ \ \ \= FIX-UP-DIRTY($v$). \\
  					04. \> \> {\sc Else if} there exists some {\sc Down-Dirty} node $v$: \\
  					05. \> \> \> FIX-DOWN-DIRTY($v$). 
  				\end{tabbing}
  			\end{minipage}
  		}}
  		\caption{\label{mainbody:fig:fix:dirty} FIX-DIRTY($\cdot$)}
  	\end{figure}

\smallskip
\noindent
The algorithms is simple -- as long as some {\sc Dirty} node remains, it runs either FIX-UP-DIRTY or FIX-DOWN-DIRTY to take care of \UD and \DD nodes respectively. One crucial aspect is that we prioritize \UD nodes over \DD ones.

\smallskip
\noindent{\bf FIX-DOWN-DIRTY($v$):}  Suppose that $\ell(v) = i$ when the subroutine is called at time $t$. By definition, we have $i > 0$ and $W_v(t) \leq 1/(\alpha \beta^2)$. 
We need to increase the value of $W_v$ if we want to ensure that $v$ no longer remains {\sc Dirty}.  This means that we should decrease the level of $v$, so that some of the hyper-edges incident on $v$ can increase their weights.
Accordingly, we find the {\em largest} possible level $j \in \{1,\ldots,(i-1)\}$ such that $W_{v\to j}(t) > 1/\beta^2$, and move the node $v$ down to this level $j$. If no such level exists, that is, if even
 $W_{v\to 1}(t) \le  1/\beta^2$, then we move the node $v$ down to level $0$. Note that in this case there is no hyper-edge $e \in E_v$ with $\ell_v(e) = 0$ for such a hyper-edge would have $w(e) = \beta^{-1} > 1/\beta^2$ when $v$ is moved to level $1$. 
 In particular, we get $W_{v\to 0}(t) = W_{v\to 1}(t)$.

 \begin{claim}
 \label{mainbody:cl:fix:down:dirty}
 \textrm{\em FIX-DOWN-DIRTY($v$)} makes the node $v$ \Clean. 
 \end{claim} 
 
\begin{proof}
	Suppose node $v$ was at level $i$ when FIX-DOWN-DIRTY($v$) was called at time $t$ and it ended up in level $j < i$.
 If $j = 0$, then $W_{v\to 0}(t) \leq 1/\beta^2$, and so $v$ becomes \Clean after time $t$. Henceforth assume $j > 0$.
Since $j\in \{1,\ldots,i-1\}$ is the maximum level where $W_{v\to j}(t) > 1/\beta^2$, we have $W_{v\to (j+1)}(t) \leq 1/\beta^2$. 
Now note that $W_{v\to j}(t) \leq \beta\cdot W_{v\to (j+1)}(t)$ since weights of hyper-edges can increase by at most a factor $\beta$ when one end point drops exactly one level.
This implies $W_{v\to j}(t) \leq 1/\beta$. Together we get that after time $t$ when  $v$ is fixed to level $j$, we have $1/\beta^2 < W_v \leq 1/\beta$.

Now we argue about the up-weights. Note that every hyper-edge $e$ that contributes to $W^+_{v\to j}(t)$ must have $\ell_v(e) \ge (j+1)$. The weight of such a hyper-edge remains unchanged as 
$v$ moves from level $(j+1)$ to $j$. We infer that $W^+_{v\to j}(t) \leq W_{v\to (j+1)}(t) \leq 1/\beta^2$. Therefore after time $t$ when $v$ is fixed at level $j$, we have $W^+_v \le 1/\beta^2.$
In sum, $v$ becomes \Clean after time $t$.
\end{proof}

\noindent{\bf FIX-UP-DIRTY($v$):}  Suppose that $\ell(v) = i$ when the subroutine is called at time $t$. At this stage, we have either \{$i = 0, W_v(t) > 1/\beta^2$\} or \{$i > 1, W_v(t) \geq 1$\}. 
We need to increase the level of $v$ so as to reduce the weight faced by it.
Accordingly, we find the {\em smallest} possible level $j \in \{i+1,\ldots, L\}$ where $W_{v\to j}(t) \leq 1/\beta$ and move $v$ up to level $j$.  Such a level $j$ always exists because $W_{v\to L}(t) \leq n^f \cdot \beta^{-L} \leq 1/\beta$.

 \begin{claim}
 \label{mainbody:cl:fix:up:dirty}
After a call to the subroutine {\em FIX-UP-DIRTY($v$)} at time $t$, we have $1/\beta^2 < W_v \leq 1/\beta$.
 \end{claim}

\begin{proof}
Suppose that the node $v$ moves up from level $i$ to level $j > i$. We now consider four possible cases.

\begin{itemize}
\item {\em Case 1. We have $i > 0$.}

Since $j \in \{i+1,\ldots,L\}$ is the minimum possible level where $W_{v\to j}(t) \leq 1/\beta$, and  since $W_{v \to i}(t) > 1/\beta$, we infer that $W_{v\to (j-1)}(t) > 1/\beta$. As the node $v$ moves up from level $(j-1)$ to level $j$, the weight it faces can drop by at most a factor of $\beta$. Hence, we get: $W_{v\to j}(t) \geq (1/\beta) \cdot W_{v\to (j-1)}(t) > 1/\beta^2$. Therefore after time $t$ when the node $v$ moves to level $j$, we have $1/\beta^2 < W_v \leq 1/\beta$.

\item {\em Case 2. We have $i = 0$, and there is an edge $e \in E_v$ with $\ell_v(e) = 0$ at time $t$.}
In this case
we have $W_{v \to i}(t) \geq 1$  in the beginning of time-step $t$, since  the edge $e \in E_v$ with $\ell_v(e) = 0$ has weight $w(e) = 1$. The rest of the proof is exactly similar to Case 1.

\item {\em Case 3. We have $i =0$, there is no edge $e \in E_v$ with $\ell_v(e) = 0$ at time $t$, and $j = 1$.}

The value of $W_v$ does not change as $v$ moves up from level $i = 0$ to level $j = 1$. Thus, we get:  $W_{v \to j}(t) = W_{v \to 0}(t) > 1/\beta^2$, for the node $v$ is {\sc Up-Dirty} at level $i =0$ at time $t$. Since the node does not move further up than level $j$, we get: $W_{v \to j}(t) \leq 1/\beta$.

\item {\em Case 4. We have $i = 0$, there is no edge $e \in E_v$ with $\ell_v(e) = 0$ at time $t$, and $j > 1$.}

Since the node $v$ does not stop at level $1$, we get: $W_{v \to 1}(t) > 1/\beta$. Hence, we infer that $j \in \{2, \ldots, L\}$ is the minimum possible level where $W_{v \to j}(t) \leq 1/\beta$. As the node $v$ moves up from level $(j-1)$ to level $j$, the weight it faces can drop by at most a factor of $\beta$. Hence, we get: $W_{v\to j}(t) \geq (1/\beta) \cdot W_{v\to (j-1)}(t) > 1/\beta^2$. Therefore, we have $1/\beta^2 < W_{v \to j}(t) \leq 1/\beta$.
\end{itemize}
\end{proof}

It is clear that if and when FIX-DIRTY() terminates, we are in a state which satisfies Invariant~\ref{mainbody:inv:node}. In the next section we show that after $T$ hyper-edge insertions and deletions, the total update time is indeed $O(f^2 \cdot T)$ and so our algorithm has $O(f^2)$-amortized update time.

  \section{Analysis of the algorithm}
  \label{sec:analysis}
  Starting from an empty graph $\mathcal{G} = (V, E)$,  fix any sequence of $T$ {\em updates}. The term ``update'' refers to the insertion or deletion of  a hyper-edge in $\mathcal{G}$. We  show that the total time taken by our algorithm to handle this sequence of updates is $O(f^2 \cdot T)$. We  also show that our algorithm has an approximation ratio of $O(f^3)$.
  
  \paragraph{Relevant counters.} We define three counters $C^{up}, C^{down}$ and $I^{down}$. The first two counters account for the time taken to update the data structures while the third accounts for the time taken to find the index $j$ 
  in both FIX-DOWN-DIRTY($v$) and FIX-UP-DIRTY($v$).
  Initially, when the input graph is empty, all the three counters are set to zero. 
  Subsequently, we increment these counters as follows.
  \begin{enumerate}
  	\item Suppose node $v$ moves from level $i$ to level $j>i$ upon a call of FIX-UP-DIRTY($v$).  Then for every hyper-edge $e \in E_v$ with $\ell_v(e) \leq j-1$, we increment  $C^{up}$ by one. 
    \item Suppose node $v$ moves from level $i$ to level $j < i$ upon a call of FIX-DOWN-DIRTY($v$). Then for every hyper-edge $e \in E_v$ with $\ell_v(e) \leq i$, we increment the value of $C^{down}$ by one. Furthermore, we increment the value of $I^{down}$ by $\beta^{i-2}/\alpha$.
  \end{enumerate}
  

  \medskip
  \noindent The next lemma upper bounds the total time taken by our algorithm in terms of the values of these counters. The proof of Lemma~\ref{lm:counter:time} appears in Section~\ref{sec:lm:counter:time}.
  
  \begin{lemma}
  	\label{lm:counter:time}
  	Our algorithm takes $\Theta(f \cdot (C^{up} + C^{down} + T) + f^2 I^{down})$ time to handle a sequence of $T$ updates.
  \end{lemma}

We will show that  $C^{up} = \Theta(f) \cdot T$ and $C^{down}+I^{down} = O(1) \cdot T$, which  will imply an amortized update time of $O(f^2)$ for our algorithm.    Towards this end, we now prove three lemmas that relate the values of these three counters.

\begin{lemma}
\label{lm:counter:new}
We have: $C^{down} \leq I^{down}$.
\end{lemma}

    \begin{lemma}
  	\label{lm:counter:down}
  	We have: $I^{down} \leq \frac{f}{\alpha - 1}\cdot (T +  C^{up})$.
  \end{lemma}

  \begin{lemma}
  	\label{lm:counter:up}
  	We have: $C^{up}  \leq 9 f \beta^2 \cdot (T+ C^{down})$.
	  \end{lemma}
	  
	    The proofs of Lemmas~\ref{lm:counter:new},~\ref{lm:counter:down} and~\ref{lm:counter:up} appear in Sections~\ref{sec:lm:counter:new},~\ref{sec:lm:counter:down} and~\ref{sec:lm:counter:up} respectively. All these three proofs use the concepts of epochs, jumps and phases as defined in Section~\ref{sec:epoch}. The main result of our paper (see Theorem~\ref{th:main}) now follows from Theorem~\ref{mainbody:th:old:result},   Lemma~\ref{lm:counter:time} and Lemma~\ref{lm:counter}.

  \begin{lemma}[Corollary to Lemma~\ref{lm:counter:new},\ref{lm:counter:down}, and~\ref{lm:counter:up}.]
  	\label{lm:counter}
  	We have: $C^{up} =  \Theta(f) \cdot T$ and $C^{down} + I^{down} = \Theta(1)\cdot T$.
  \end{lemma}
  \begin{proof}
Replacing $C^{down}$ in the RHS of Lemma~\ref{lm:counter:up} by the upper bounds from Lemmas~\ref{lm:counter:new} and~\ref{lm:counter:down}, we get:
\begin{eqnarray*}
C^{up} & \le & (9f\beta^2) \cdot T + (9f\beta^2) \cdot C^{down} \\
& \leq & (9f\beta^2) \cdot T + (9f\beta^2) \cdot I^{down} \\
& \leq & (9f\beta^2) \cdot T + \frac{(9f\beta^2) f}{(\alpha - 1)} \cdot (T+ C^{up}) \\
& \leq & (9f\beta^2) \cdot T + (1/4) \cdot T + (1/4) \cdot C^{up} \qquad \text{(see equation~\eqref{eq:parameter})}
\end{eqnarray*}
Rearranging the terms in the above inequality, we get: $(3/4) \cdot C^{up} \leq (9f \beta^2 + 1/4) \cdot T = (36f\beta^2+1) \cdot (T/4)$. Multiplying both sides by $(4/3)$, we get: $C^{up} \leq (12f\beta^2+1/3) \cdot T \leq (13f\beta^2)  T$. Since $\beta = 17$, we get:
\begin{equation}
\label{eq:new:900}
C^{up} \leq \Theta(f) \cdot T
\end{equation}
Since $\alpha = \Theta(f^2)$, Lemmas~\ref{lm:counter:new} and~\ref{lm:counter:down} and equation~\eqref{eq:new:900} imply that:
\begin{equation}
\label{eq:new:901}
C^{down} \leq I^{down} \leq \Theta(1) \cdot T
\end{equation}
  \end{proof}

 \subsection{Epochs, jumps and phases}
 \label{sec:epoch}
 Fix any node $v \in V$. An {\em epoch} of this node is a maximal time-interval during which the node stays at the same level. 
 An epoch ends when either (a) the node $v$ moves up to a higher level due to a call to FIX-UP-DIRTY, or (b) the node $v$ moves down to a lower level due to a call to the subroutine FIX-DOWN-DIRTY. These events are called {\em jumps}. Accordingly, there are {\sc Up-Jumps} and {\sc Down-Jump}s.  
 Next, we define a {\em phase}  of a node to be a maximal sequence of consecutive epochs where the levels of the node keep on increasing. The phase of a node $v$ is denoted by $\Phi_v$. Suppose that a phase $\Phi_v$ consists of $k$ consecutive epochs of $v$ at levels $i_1, \ldots, i_k \in \{0,1,\ldots, L\}$. Then we have: $i_1 < i_2 < \cdots < i_k$. 
By definition, the epoch immediately before $\Phi_v$ must have level larger than $i_1$ implying FIX-DOWN-DIRTY($v$) landed $v$ at level $i_1$. Similarly, the epoch subsequent to $i_k$ is smaller than $i_k$ implying FIX-DOWN-DIRTY($v$) is called again.

 \subsection{Proof of Lemma~\ref{lm:counter:new}}
 \label{sec:lm:counter:new}
 
 Suppose that a node $v$ moves down from (say) level $j$ to level $i < j$ at time (say) $t$ due to a call to the subroutine FIX-DOWN-DIRTY$(v)$. Let $\Delta^{down}$ and $\Delta_I^{down}$ respective denote the increase in the counters $C^{down}$ and $I^{down}$ due to this event. We will show that $\Delta^{down} \leq \Delta_I^{down}$, which will conclude the proof of the lemma. By definition, we have:
 \begin{equation}
 \label{eq:D:1}
 \Delta_{I}^{down} = \beta^{i-2}/\alpha
 \end{equation}
 Let $X = \{ e \in E_v : \ell_v(e) \leq i \}$ denote the set of hyper-edges incident on $v$ that contribute to the increase in $C^{down}$ due to the {\sc Down-Jump} of $v$ at time $t$. Specifically, we have: $|X| = \Delta^{down}$. Each edge $e \in X$ contributes a weight $\beta^{-i}$ towards the node-weight $W_{v \to i}(t)$. Thus, we get: $|X| \cdot \beta^{-i} \leq W_{v \to i}(t) \leq 1/(\alpha \beta^2)$. The last inequality holds since $v$ is {\sc Down-Dirty} in the beginning of time-step $t$. Rearranging the terms, we get: $\Delta^{down} = |X| \leq \beta^{i-2}/\alpha$. The lemma now follows from equation~\eqref{eq:D:1}.

 \subsection{Proof of Lemma~\ref{lm:counter:down}}
 \label{sec:lm:counter:down}
 
Suppose we call FIX-DOWN-DIRTY($v$) at some time $t_2$. Let $\ell(v) = i$ just before the call, and let $[t_1,t_2]$ be the epoch with level of $v$ being $i$.
Let $X := \{e \in E_v : \ell_v(e) \leq i\}$ at time $t_2$. By definition, $I^{down}$ increases by $\beta^{i-2}/\alpha$ during the execution of FIX-DOWN-DIRTY($v$); let us call this increase $\Delta_I^{down}$. Thus, we have:
  \begin{equation}
 \label{eq:x}
 \Delta_I^{down} = \beta^{i-2}/\alpha
 \end{equation}
 Consider the time between $[t_1,t_2]$ and let us address how $W_v$ can decrease in this time while $v$'s level is fixed at $i$. Either some hyper-edge incident on $v$ is deleted, or some hyper-edge $e \in E_v$ incident on it decreases its weight. In the latter case, the level $\ell(e)$ of such an hyper-edge $e$ must increase above $i$.
 Let $\Delta^T$ denote the number of hyper-edge deletions incident on $v$ during the time-interval $[t_1, t_2]$. 
 Let $\Delta^{up}$ denote the increase in the value of  $C^{up}$ during the time-interval $[t_1, t_2]$ {\em due to the hyper-edges incident on $v$}. 
 Specifically, at time $t_1$, we have $\Delta^T = \Delta^{up} = 0$. Subsequently, during the time-interval $[t_1, t_2]$, we increase the value of $\Delta^{up}$ by one each time we observe that a hyper-edge $e \in E_v$  increases its level $\ell(e)$ to something larger than $i$. 
 Note that $\ell(v) = i$ throughout the time-interval $[t_1, t_2]$. 
 Hence, each time we observe an unit increase in $\Delta^T + \Delta^{up}$, this decreases the value of $W_v$ by at most $\beta^{-i}$. Just before time $t_1$, the node $v$ made either an {\sc Up-Jump}, or a {\sc Down-Jump}. 
 Hence, Claims~\ref{mainbody:cl:fix:up:dirty} and~\ref{mainbody:cl:fix:down:dirty} imply that $W_{v\to i}(t_1) > 1/\beta^2$. 
 As $W_v(t_2) \leq 1/(\alpha \beta^2)$ at time $t_2$, we infer that  $W_v$ has dropped by at least $(1-1/\alpha) \cdot \beta^{-2}$ during the time-interval $[t_1, t_2]$.  In order to account for this drop in $W_v$, the value of $\Delta^T + \Delta^{up}$ must have increased by at least $(1-1/\alpha) \cdot \beta^{-2}/\beta^{-i} = (1-1/\alpha) \cdot \beta^{i-2}$. Since $\Delta^T = \Delta^{up} = 0$ at time $t_1$, at time $t_2$ we get: $\Delta^T + \Delta^{up} \geq (1-1/\alpha) \cdot \beta^{i-2}$.  Hence, \eqref{eq:x} gives us:
 \begin{equation}
 \label{eq:x:1}
 \Delta_I^{down} \leq (\alpha -1)^{-1} \cdot (\Delta^T + \Delta^{up})
 \end{equation}
 
 Each time the value of $I^{down}$ increases due to FIX-DOWN-DIRTY on some node, inequality~\eqref{eq:x:1} applies. 
 If we sum all these inequalities, then the left hand side (LHS) will be exactly equal to the final value of $I^{down}$, and the right hand side (RHS) will be at most $(\alpha -1 )^{-1} \cdot (f \cdot T + (f-1) \cdot C^{up})$. The factor $f$ appears in front of $T$ because each hyper-edge deletion can contribute $f$ times to the sum $\sum \Delta^T$, once for each of its endpoints. Similarly, the factor $(f-1)$ appears in front of $C^{up}$ because whenever the level of an hyper-edge $e$ moves up due to the increase in the level $\ell(v)$ of some endpoint $v \in e$, this contributes at most $(f-1)$ times to the sum $\sum \Delta^{up}$, once for every other endpoint $u \in e, u \neq v$. Since LHS $\leq$ RHS, we get: $I^{down} \leq (\alpha - 1)^{-1} \cdot (f \cdot T + (f-1) \cdot C^{up}) \leq  (f/(\alpha - 1)) \cdot (T+C^{up})$. This concludes the proof of the lemma.

 \subsection{Proof of Lemma~\ref{lm:counter:up}}
 \label{sec:lm:counter:up}
Fix a node $v$ and consider a phase $\Phi_v$ where $v$ goes through levels $i_1 < i_2 < \cdots < i_k$. Thus, the node  $v$ enters the level $i_1$ at time $t_1$ (say) due to a call to FIX-DOWN-DIRTY$(v)$. For $r \in [2, k]$, the node $v$ performs an {\sc Up-Jump} at time $t_r$ (say) from the level $i_{r-1}$ to the level $i_r$, due to a call to FIX-UP-DIRTY$(v)$. This implies that $t_1 < t_2 < \cdots < t_k$. The phase ends, say, at time $t_{k+1} > t_k$  when the node $v$ again performs a {\sc Down-Jump} from the level $i_k$ due to a call to FIX-DOWN-DIRTY$(v)$. 

Let $\Delta^{up}$ denote the total increase in the value of the counter $C^{up}$ due to  the phase $\Phi_v$. For $r \in [2, k]$, let $\Delta^{up}_r$ denote the increase in the value of the counter $C^{up}$ due to the {\sc Up-Jump} of $v$ at time $t_r$. Thus, we have: 
\begin{equation}
\label{eq:delta:new}
\Delta^{up} = \sum_{r = 2}^k \Delta^{up}_r
\end{equation}

We define two more counters: $\Delta^T$, $\Delta^{down}$. The former counter equals the number of hyper-edge insertions/deletions incident on $v$ during the time-interval $[t_1, t_k]$. The latter counter equals the increase in the value of $C^{down}$ due to the hyper-edges incident on $v$ during the time-interval $[t_1, t_k]$. Alternately, these two counters can be defined as follows. At time $t_1$, we set $\Delta^T \leftarrow 0$ and $\Delta^{down} \leftarrow 0$. Subsequently, whenever at any time $t \in [t_1, t_k]$, a hyper-edge incident on $v$ gets inserted into or deleted from the input graph, we increment the value of $\Delta^T$ by one. Further, whenever at any time $t \in [t_1, t_k]$, a hyper-edge $e$ incident on $v$ gets its level decreased because of a {\sc Down-Jump} of some node $u \in e, u \neq v$, we increment the value of $\Delta^{down}$ by one. 

Since $v$ enters the level $i_1$ at time $t_1$ due to a call to FIX-DOWN-DIRTY$(x)$,  Claim~\ref{mainbody:cl:fix:down:dirty} implies that:
 \begin{equation}\label{eq:007}
 W_{v\to i_1}({t_1}) \leq 1/\beta ~~~~~~~~~~ \textrm{ and } ~~~~~~~~~~~ W^+_{v\to i_1}({t_1}) \leq 1/\beta^2
 \end{equation}
 
Our main goal is to upper bound $\Delta^{up}$ in terms of the final values of the counters $\Delta^T$ and $\Delta^{down}$. 

\begin{claim}\label{clm:006}
For $2\le r \le k$, $\Delta^{up}_r \leq \beta^{i_r - 1}$.
\end{claim}
\begin{proof}
	By Claim~\ref{mainbody:cl:fix:up:dirty} we have  $W_{v\to i_r}(t_r) \leq 1/\beta$, that is, the total weight incident on $v$ after it has gone through FIX-UP-DIRTY at time $t_r$ is at most $1/\beta$. Now, each hyper-edge $e \in E_v$ which contributes to $\Delta^{up}_r$ has weight, right after time ${t_r}$, precisely $\beta^{-\ell(v)} = \beta^{-i_r}$.
	Putting together, we get $\Delta^{up}_r \leq \beta^{i_r - 1}$.
\end{proof}

Using the above claim, we get the following upper bound on the sum of all but the last $\Delta^{up}_k$.
\begin{claim}\label{clm:007}
We have: $\sum_{r=2}^{k-1} \Delta^{up}_r < 2 (\Delta^T + \Delta^{down})$.
\end{claim}

\begin{proof}
If $k = 2$, then we have an empty sum $\sum_{r = 2}^{k-1} \Delta^{up}_r = 0$, and hence the claim is trivially true. For the rest of the proof, we suppose that $k > 2$, which implies that $i_{k-1} \geq i_2 > i_1 \geq 0$. Thus, we get:
\begin{equation}
\label{eq:donald:trump}
k > 2 \text{ and } i_{k-1} > 0.
\end{equation}
 
 To continue with the proof, summing over the inequalities from Claim~\ref{clm:006}, we get:
\begin{equation}
\label{eq:new:110}
\sum_{r = 2}^{k-1} \Delta^{up}_r \leq \beta^{i_{k-1}}
\end{equation}

Since the node  $v$ performs an {\sc Up-Jump} at time $t_{k-1}$ from level $i_{k-1} > 0$ (see equation~\eqref{eq:donald:trump}), the node must be \UD at that time. It follows that $W_{v\to i_{k-1}}(t_{k}) > 1$. From equation~\eqref{eq:007}, we have 
$W_{v \to i_{k-1}}(t_1) \leq W_{v\to i_{1}}(t_1) < 1/\beta$. Thus, during the time interval $[t_1,t_k]$ the value of $W_{v \to i_{k-1}}$
increases by at least $(1-1/\beta)$.
This can be either due to (a) some hyper-edge incident to $v$ being inserted, or (b) some hyper-edge $e \in E_v$ gaining its weight because of some endpoint $u \in e, u \neq v,$ going down. The former increases $\Delta^T$ and the latter increases  $\Delta^{down}$.
Furthermore, the increase in $W_{v \to i_{k-1}}$ due to every such hyper-edge  is at most $\beta^{-i_{k-1}}$. This gives us the following lower bound.
\begin{equation}
\label{eq:008}
\Delta^T + \Delta^{down} \geq (1 - 1/\beta)\cdot \beta^{i_{k-1}} > \beta^{i_{k-1}}/2 ~~ \textrm{ and so from equation~\eqref{eq:new:110},  we have}~~~~ \sum_{r=2}^{k-1} \Delta^{up}_r < 2(\Delta^T + \Delta^{down})
\end{equation}
The claim follows from equation~\eqref{eq:008}.
\end{proof}

It now remains to upper bound  $\Delta^{up}_k$. This is done in Claim~\ref{cl:A:500}, whose proof appears in Section~\ref{sec:new:claim}.

\begin{claim}
\label{cl:A:500}
We have: $\Delta^{up}_k < (8\beta^2) \cdot (\Delta^T + \Delta^{down})$.
\end{claim}
From equations~\ref{eq:parameter},~\ref{eq:delta:new} and Claims~\ref{clm:007},~\ref{cl:A:500},  we get:
\begin{equation}\label{eq:case1}
\Delta^{up} \le (8\beta^2 + 2)(\Delta^T + \Delta^{down}) \leq (9\beta^2) \cdot (\Delta^T + \Delta^{down})
\end{equation}
Using equation~\eqref{eq:case1}, now we can prove Lemma~\ref{lm:counter:up}. 
For every phase of a node $v$, as per equation~\eqref{eq:case1} we can charge the increase in $C^{up}$ to the increase in $(T + C^{down})$ corresponding to hyper-edges incident of $v$. Summing up over all nodes and phases, the LHS gives $C^{up}$ while the RHS gives $(9\beta^2) \cdot (f \cdot T + (f-1) \cdot C^{down})$. The coefficient $f$ before $T$ comes from the fact that every hyper-edge insertion can contribute $f$ times to the RHS, once for each of its endpoints. The coefficient $(f-1)$ before $C^{down}$ comes from the fact that whenever the level of a hyper-edge $e$ decreases due to the {\sc Down-Jump} of a node $u \in e$, this event contributes at most $(f-1)$ times to the RHS: once for every other endpoint $v \in e, v \neq u$. Thus, we get:
$$C^{up} \le (9\beta^2) \cdot (f \cdot T + (f-1) \cdot C^{down})  \leq 9f\beta^2 \cdot (T+ C^{down})$$

\subsubsection{Proof of Claim~\ref{cl:A:500}}
\label{sec:new:claim}
We fork into two cases. \smallskip

\paragraph{Case 1:} $i_k \leq i_{k-1} + 3$.

\medskip
\noindent
 Since $i_k \leq i_{k-1}+3$, Claim~\ref{clm:006} implies that:
 \begin{equation}
 \label{eq:B:new} 
 \Delta^{up}_k \leq \beta^{i_k - 1} \leq \beta^{i_{k-1}} \cdot \beta^2
 \end{equation}
 In the proof of Claim~\ref{clm:007}, we derived equation~\eqref{eq:008}, which gives us:
 \begin{equation}
 \label{eq:B:old}
 \Delta^T + \Delta^{down} \geq \beta^{i_{k-1}}/2
 \end{equation}
 Equations~\ref{eq:B:new},~\ref{eq:B:old} imply that $\Delta^{up}_k < 2\beta^2 \cdot (\Delta^T + \Delta^{down}) \leq 8\beta^2 \cdot (\Delta^T + \Delta^{down})$. This concludes the proof of Claim~\ref{cl:A:500} for this case.

\paragraph{Case 2:} $i_k > i_{k-1} + 3$.

\newcommand{\W}{\mathcal{W}}
\renewcommand{\t}{t^*}

\medskip
\noindent
We start by noting that the weight of a node is always less than $2$ at every time.
\begin{claim}
\label{cl:new:002}
We have: $W_v(t) < 2$ at every time $t$. 
\end{claim}

\begin{proof}
The crucial observation is that fixing an {\sc Up-Dirty} node $u$ never increases the weight of any node.  Furthermore,  a {\sc Down-Dirty} node gets fixed only if no other node is {\sc Up-Dirty} (see Figure~\ref{mainbody:fig:fix:dirty}).

In the beginning of time-step $t = 0$, the input graph is empty, and we clearly have $W_v(t) = 0 < 1$. By induction, suppose that $W_v(t) < 1$ in the beginning of  some time-step $t$.  Now, during time-step $t$, the weight $W_v$ can increase only if one of the following events occur: 
\begin{itemize}
\item (a) A hyper-edge containing $v$ gets inserted into the graph. This increase the value of $W_v$ by at most one. Thus, we have $W_v(t+1) < 2$. 
\item (b) We call the subroutine FIX-DOWN-DIRTY$(u)$ for some node $u$. Note that  fixing a {\sc Down-Dirty} node $u$ can increase the weight $W_u$ by at most one, and hence this can increase the weight of a neighbour of $u$  also by at most one. It again follows that $W_v(t+1) < 2$. 
\end{itemize}
If $W_v(t+1) < 1$, then we are back in the same situation as in time-step $t$. Otherwise, if $1 \leq W_v(t+1) < 2$, then $v$ is {\sc Up-Dirty} in the beginning of time-step $t+1$. In this case,  no {\sc Down-Dirty} node gets fixed (and no hyper-edge gets inserted) until we ensure that $W_v$ becomes smaller than one. Hence, the value of $W_v$ always remains smaller than $2$.
\end{proof}

\begin{claim}
\label{cl:new:simple}
We have: $W_{v \to (i_{k} - 1)} (t_k) > 1/\beta$.
\end{claim}

\begin{proof}
While making the {\sc Up-Jump} at time $t_k$, the node $v$ does not stop at level $i_k -1$.  The claim follows.
\end{proof}

\begin{claim}
\label{cl:new:simple:1}
We have: $W_{v \to (i_k - 3)}^+(t_k) > 1/(2\beta)$.
\end{claim}

\begin{proof}
Suppose that the claim does not hold. Then we get:
\begin{equation}
\label{eq:009}
W_{v\to (i_k - 1)}(t_k) ~~ \leq ~~ W^+_{v\to (i_k - 3)}(t_k) + \frac{W_{v\to (i_k - 3)}(t_k) - W^+_{v\to (i_k - 3)}(t_k)}{\beta^2} ~~\leq~~ 1/(2\beta) + 2/\beta^2 \leq 1/\beta 
\end{equation}
The first inequality holds since the weights of the hyper-edges $e \in E_v$ with $\ell_v(e) \leq i_k - 3$ get scaled by at least a factor of $1/\beta^2$ when $v$ moves from level $i_k - 3$ to $i_k - 1$, and the rest 
can only go down. The second inequality  holds since $W_{v\to (i_k - 3)}(t_k) \leq W_{v \to i_{k-1}}(t_k) < 2$ by Claim~\ref{cl:new:002} and the assumption $W^+_{v\to (i_k - 3)}(t_k) \leq 1/(2\beta)$.
The last inequality holds since $\beta = 17$ by equation~\eqref{eq:parameter}. However, equation~\eqref{eq:009} contradicts Claim~\ref{cl:new:simple}.
\end{proof}

Claim~\ref{cl:new:simple:1} states that $W^+_{v\to (i_k - 3)}(t_k) > 1/(2\beta)$. Since  $i_{k-1} < i_k - 3$, equation~\eqref{eq:007} implies  that $W^+_{v\to (i_k - 3)}(t_1) \leq W^+_{v\to i_{k - 1}}(t_1)\leq W^+_{v \to i_{1}}(t_1) \leq 1/\beta^2$.
Thus during the time-interval $[t_1,t_k]$, the value of $W^+_{v\to (i_k - 3)}$ increases by at least $1/(2\beta) - 1/\beta^2$. This increase can occur in three ways:
(1) a hyper-edge $e \in E_v$ is inserted with $\ell_v(e) > i_k - 3$ before the {\sc Up-Jump} at time $t_k$
(which contributes to $\Delta^T$), (2) some hyper-edge $e \in E_v$ gains weight due to a {\sc Down-Jump} of some node (say) $u \in e, u \neq v$,  and $\ell_v(e) > i_k - 3$ after the {\sc Down-Jump} (which contributes to $\Delta^{down}$), and (3) some hyper-edge $e \in E_v$ had $\ell_v(e) \le i_k - 3$ at $t_1$ but $\ell_v(e) > i_k - 3$ at $t_k$. Note that the total weight of the hyper-edges of type (3) at time $t_1$ incident on $v$ at level $i_k - 3$ is at most  $1/\beta$; this follows from \eqref{eq:007}. Therefore, when $\ell_v(e)$ for such an edge $e$ raises to $> i_k - 3$, the weight decreases by at least a $1/\beta$ factor. Hence the total increase in $W^+_{v \to (i_k - 3)}$
due to type (3) hyper-edges is at most $1/\beta^2$, and the weight increase of at least $1/(2\beta) - 2/\beta^2$ must come from hyper-edges of type (1) and type (2). However each such hyper-edge $e$ can contribute at most $\beta^{-(i_k - 2)}$ 
to the weight (since  $\ell_v(e) > i_k - 3$). And therefore, we get (recall that $\beta = 17$ by equation~\eqref{eq:parameter}):
\[
\left(\Delta^T + \Delta^{down}\right)\cdot \beta^{-(i_k - 2)} \geq \frac{1}{2\beta} - \frac{2}{\beta^2} > \frac{1}{8\beta} ~~~~ \textrm{ implying } ~~~~ \left(\Delta^T + \Delta^{down}\right) \ge \frac{\beta^{i_k - 1}}{8\beta^2}
\]
Claim~\ref{clm:006} gives $\Delta^{up}_k \leq \beta^{i_k - 1}$, and therefore we get: 
\begin{equation}\label{eq:case2}
\Delta^{up}_k \leq 8\beta^2\left(\Delta^T + \Delta^{down}\right)
\end{equation}
This concludes the proof of Claim~\ref{cl:A:500} for this case.

  \subsection{Proof of Lemma~\ref{lm:counter:time}}
  \label{sec:lm:counter:time}
  For technical reasons, we assume that we end with the empty graph as well. This is without loss of generality due to the following reason. Suppose we made $T$ updates and the current graph is $G$. At this point, the graph has $T' \leq T$ edges.
  Suppose the time taken by our algorithm till now is $T_1$. Now delete all the $T'$ edges, and let the time taken by our algorithm to take care of these $T'$ updates be $T_2$. If $T_1 + T_2 = \Theta(f^2(T+T')) = \Theta(f^2T)$, then $T_1 = \Theta(T)$ as well.
  Therefore, we assume we end with an empty graph.
  
    	When a hyper-edge $e$ is inserted into or deleted from the graph, we take $O(f)$ time to update the relevant data structures for its $f$ endpoints. The rest of the time is spent in implementing the {\sc While} loop in Figure~\ref{mainbody:fig:fix:dirty}. 
    	We take care of the two subroutines separately.
    	
    		\smallskip
    		\noindent {\em Case 1.}  The subroutine FIX-DOWN-DIRTY($v$) is called which moves the node $v$ from level $i$ to level $j < i$ (say). 
    		We need to account for the time to find the relevant index $j$ and the time taken to update the relevant data structures.
    		By Lemma~\ref{lm:time:move:down}, the time taken for the latter is proportional $\Theta(f\cdot \Delta C^{down})$. Further, the value of  $C^{up}$ remains unchanged. 
    		For finding the index $j < i$, it suffices to focus on the edges $E_{v,i} = \{e\in E_v: \ell_v(e) \le i\}$ since these are the only edges that change weight as $v$ goes down. Therefore, this takes time $\Theta(|\{e\in E_v: \ell_v(e) \le i\}|)$.
    		Since each of these edges had $w(e) = \beta^{-i}$ and since $W_v \le \frac{1}{\alpha\beta^2}$ before the FIX-DOWN-DIRTY($v$) call, we have $|\{e\in E_v: \ell_v(e) \le i\}| \le \beta^{i-2}/\alpha$ which is precisely $\Delta I^{down}$.
    		Therefore, the time taken to find the index $j$ is $\Theta(\Delta I^{down})$.
    		
  	\smallskip
  	\noindent {\em Case 2.} The subroutine FIX-UP-DIRTY($v$) is called which moves the node $v$ from level $i$ to level $j > i$, say. 
  	Once again, we need to account for the time to find the relevant index $j$ and the time taken to update the relevant data structures, and once again by Lemma~\ref{lm:time:move:up} the time taken for the latter is $\Theta(f\cdot \Delta C^{up})$. Further, the value of $C^{down}$ remains unchanged. 
  	We now account for the time taken to find the index $j$. 
  	\begin{claim}\label{ckm:hillary}
  		$j$ can be found in time $\Theta(j-i)$.
  	\end{claim}
  	\begin{proof}
  		To see this note that for $k\geq i$, 
  		\[
\textstyle   		W_{v\to k}(t) =  \sum_{\ell \geq k} \sum_{e\in E_{v,\ell}} w(e) + \sum_{\ell < k} \frac{1}{\beta^{k-\ell}} \sum_{e\in E_{v,\ell}} w(e)
  		\]
  		since (a) edges not incident on $v$ are immaterial, (b) the edges incident on $v$ whose levels are already $\geq k$ do not change their weight, and (c) edges whose levels are $\ell < k$ have their weight go from $\beta^{-\ell}$ to $\beta^{-k}$.
  		The above implies that for $k\geq i$, 
  		\[
\textstyle   		W_{v\to (k+1)} (t) = W_{v\to k} (t) - \left(1-\frac{1}{\beta}\right)\sum_{e\in E_{v,k}} w(e) = W_{v\to k} (t) - \left(1-\frac{1}{\beta}\right)|E_{v,k}|\cdot \beta^{-k}
  		\]
  		That is, $W_{v\to(k+1)}$ can be evaluated from $W_{v\to k}(t)$ in $\Theta(1)$ time since we store $|E_{v,k}|$ in our data structure. The claim follows.
  	\end{proof}
  	
  	Note that the LHS of Claim~\ref{ckm:hillary} can be as large as $\Theta(\log n)$. To account for the movement, we again fix a vertex $v$ and a phase $\Phi_v$ where the level of $v$ changes from $i_1$ to say $i_k$. The total time for finding indices is $\Theta(i_k - i_1)$.
  	After this, there must be a DOWN-JUMP due to a call to FIX-DOWN-DIRTY($v$) since the final graph is empty. Thus, we can charge the time taken in finding indices in this phase $\Phi_V$ to $\Delta I^{down}$ in the FIX-DOWN-DIRTY($v$) call right at the end of this phase.
  	We can do so since $\Delta I^{down} = \beta^{i_k - 2}/\alpha = \frac{1}{f^2} \Theta(i_k)$ since $\beta = \Theta(1)$ and $\alpha = \Theta(f^2)$ by \eqref{eq:parameter}. Therefore, the total time taken to find indices in the FIX-UP-DIRTY($v$) calls in all is at most $f^2 I^{down}$.
\smallskip

\noindent
In sum, the total time taken to initialize update data structures is at most $\Theta\left(f\cdot \left(C^{up} + C^{down} + T\right)\right)$ and the total time taken to find indices is at most $\Theta(f^2\cdot I^{down})$. This proves Lemma~\ref{lm:counter:time}.

\bibliographystyle{abbrv}
\bibliography{paper}

\end{document}